\documentclass[conference]{IEEEtran}
\IEEEoverridecommandlockouts
\usepackage{cite}
\usepackage{amsmath,amssymb,amsfonts}
\usepackage{filecontents}
\usepackage[linesnumbered,ruled,vlined]{algorithm2e}
\usepackage{graphicx}
\usepackage{textcomp}
\usepackage{amsthm}
\usepackage{xcolor}
\usepackage{url}

\def\BibTeX{{\rm B\kern-.05em{\sc i\kern-.025em b}\kern-.08em
    T\kern-.1667em\lower.7ex\hbox{E}\kern-.125emX}}
    
\begin{document}

\title{A Simple and Intuitive Algorithm for Preventing Directory Traversal Attacks}

\author{\IEEEauthorblockN{Michael Flanders}
\IEEEauthorblockA{The University of Texas at Austin}
\textit{Undergraduate, Electrical and Computer Engineering}\\
Austin, TX \\
flanders.michaelk@utexas.edu
}

\maketitle

\begin{abstract}
With web applications becoming a preferred method of presenting graphical user interfaces to users, software vulnerabilities affecting web applications are becoming more and more prevalent and devastating. Some of these vulnerabilities, such as directory traversal attacks, have varying defense mechanisms and mitigations that can be difficult to understand, analyze, and test. Gaps in the testing of these directory traversal defense mechanisms can lead to vulnerabilities that allow attackers to read sensitive data from files or even execute malicious code.

This paper presents an analysis of some currently used directory traversal attack defenses and presents a new, stack-based algorithm to help prevent these attacks by safely canonicalizing user-supplied path strings. The goal of this algorithm is to be small, easy to test, cross-platform compatible, and above all, intuitive. We provide a proof of correctness and verification strategies using symbolic execution for the algorithm. We hope that the algorithm is simple and effective enough to help move developers towards a unified defense against directory traversal attacks.

\end{abstract}


\section{Introduction}
In the age of Web 2.0, web application usage is on the rise, and so are web application security vulnerabilities. Between 2016 and 2017, the number of vulnerabilities published to the National Vulnerability Database (NVD) increased by 127\% with web application vulnerabilities making up 51\% of all disclosed vulnerabilities for 2017 \cite{yearlyreport}. Despite a large amount of research and literature on preventing certain vulnerability classes that plague web applications--e.g. SQL injection and cross-site scripting \cite{pixy, sqli, xss}--there does not seem to be much literature on directory traversal vulnerabilities which made up 22\% of 2017’s vulnerability disclosures \cite{yearlyreport}. Furthermore, a lot of material on directory traversal attacks simply recommend using web application scanners to detect these vulnerabilities. There are two problems with this approach: finding vulnerabilities is not the same as fixing them, and many scanners using default configurations fail to detect traversal vulnerabilities \cite{Doupe}.

More effective approaches to preventing directory traversal attacks include whitelisting \cite{directorytraversallecture, SurveySymExec-CSUR18, whitelisting}, removing traversal characters from path strings \cite{thttpd, minihttpd, symlinks}, and canonicalizing path strings to obtain an absolute file name \cite{coreutils, phprealpath, notgethacked}. Trying to prevent these attacks using any one of these solutions presents challenges. For example, whitelisting-only approaches leave holes in defenses when paths contain symbolic links, methods removing traversal characters are often semantically incorrect or have memory-safety issues, and path string canonicalization relies on the filesystem and is often not cross-platform or lightweight.

This paper presents an algorithm that prevents directory traversal attacks and solves the issues presented by the aforementioned mitigations. We use both whitelisting and path string canonicalization to help create a solid defense against directory traversal attacks. The algorithm can handle path strings containing symbolic links and is cross-platform, lightweight, intuitive, and simple to verify using symbolic execution.

More background information on path string semantics, directory traversal vulnerabilities, and currently used mitigations is presented in Section 2. We present our algorithm and a proof-of-correctness in Section 3. We discuss verification strategies for implementations of our algorithm in Section 4. Section 5 details advantages, concerns, and concluding remarks.

\section{Background}

\subsection{Directory Structure and Traversal}

A path string is a string that denotes the location of a file or directory on a system’s file system. Fig. 1 shows a minimal example of a hierarchical file system for a Linux system composed of the root directory, "/",  and the "www" and "home" directories that are housed in the root directory.  The user folder for the user named "alice" can be represented by the path string, "/home/alice/"; this path string starts with the root directory, "/", traverses through the "home" directory, and then selects the "alice" directory inside. 

\begin{figure}[htp]
    \centering
    \includegraphics[width=6cm]{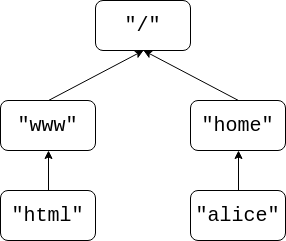}
    \caption{Example directory structure for a Linux system}
    \label{fig:directory_structure_example}
\end{figure}

Some file systems include special characters and directories that allow users to easily traverse this hierarchy. Not represented in Fig. 1 are the hidden "." and ".." directories contained inside each directory. These special names refer to directories relative to the directory that they are inside; the ".." directory really refers to the parent directory of the current directory, and the "." directory refers to the current directory. For example, if we are inside of the "/home/alice" directory, then ".." refers to "/home/" and "." refers to "/home/alice/". 

Attackers routinely use these special directory identifiers to traverse file hierarchies to reach restricted directories. For example, if a program is supposed to use a user-supplied path string to retrieve a file from inside Alice’s home folder, but the user instead supplies a path string of "../../www/html/", then they can potentially access restricted files in other branches of the file system. This is what is commonly known as a directory traversal or path traversal attack.

\subsection{Directory Traversal Vulnerabilities}
Directory traversal vulnerabilities are caused by a program using a user-supplied path string to fetch or download files without checking that a path string refers to the correct file. Web applications are especially vulnerable to directory traversal issues as web apps routinely retrieve and store user-uploaded files on the web application’s file system. 

\begin{figure}[htp]
    \centering
    \includegraphics[width=7cm]{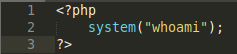}
    \caption{Example malicious PHP file contents}
    \label{fig:example_php_shell}
\end{figure}

For example, consider a PHP web forum written that allows site members to upload profile pictures as jpg files to the endpoint, "myforum.net/userprofile.php?picture=profile.jpg". Suppose instead that a user uploads the code shown in Fig. 2 with a filename of "../../shell.php" or with however many "../" tokens as required to reach the web application’s code directory. Then, an attacker can make a request to this endpoint--i.e. "myforum.net/shell.php"--to execute arbitrary PHP code on the webserver.

\subsection{Currently Used Mitigations}

\subsubsection{Removing Traversal Characters}
Approaches that try to prevent directory traversal attacks by removing substrings such as "../" from the path string normally do so in order to allow flexibility in file selection while blocking a small amount of blacklisted security-critical files. Take for example mini\_httpd--a fork-based HTTP server used to serve the ‘http://routerlogin.net’ page on some NETGEAR routers \cite{minihttpd}--and thttpd--a small, fast, and secure HTTP server that claims to power some popular sites such as images.paypal.com, garfield.com, drudgereport.com, the official website of the Sovereign Principality of Sealand, etc \cite{thttpd}. Both of these HTTP servers are developed and maintained by the same company, ACME Labs \cite{minihttpd, thttpd}, so both use the de\_dotdot algorithm shown in Fig. 3 to remove directory traversal characters from user-supplied path strings. From looking through the source code of both HTTP servers, we presume that the HTTP servers remove directory traversal characters in order to blacklist specific files such as ".htpasswd" and "/etc/passwd" while allowing maximum flexibility in accessing most other files in the current HTTP server’s directory.

While the HTTP servers seem to have no directory traversal related vulnerabilities because of the code shown in Fig. 3, this algorithm, like many other algorithms that try this approach, contains semantic errors. If assuming that the function’s name, de\_dotdot, implies only that all ".." tokens are removed from the pathstring, then the function is semantically incorrect as the path string "/etc../" contains ".." tokens and is untouched by the sanitizations. On the other hand, if assuming that this function means to prevent directory traversal, it is still incorrect as it does not consider path strings starting from the root directory--e.g. "/etc/passwd"--that traverse the file system bottom-up and still resulting in directory traversal vulnerabilities\footnote{We believe that this flaw has been fixed by the most recent patch, version 1.30, to mini\_httpd that addresses CVE-2018-18778 \cite{NVD:softonline, minihttpd}. However, the fix for preventing bottom-up-directory-traversal attacks has been applied outside of the de\_dotdot function meaning that the de\_dotdot method is still semantically incorrect.}.

The problem with these types of approaches is that they are difficult to test and verify. Testers cannot think of all possible malicious path strings to test against, and verification by formal methods or symbolic execution face path explosion problems with the number of string operations contained inside of unbounded loops \cite{SurveySymExec-CSUR18}. Holes in test coverage of algorithms adopting this approach can end up missing directory traversal vulnerabilities causing vulnerabilities such as CVE-2018-18778 in mini\_httpd and thttpd \cite{NVD:softonline, minihttpd, thttpd}.

\begin{figure}[]
\centering
\includegraphics[width=\linewidth]{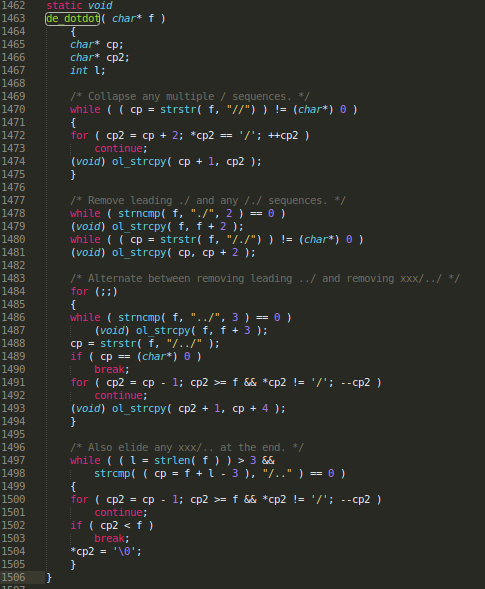}
\caption{Code to remove traversal characters from user-supplied path strings in mini\_httpd and t\_httpd \cite{minihttpd, thttpd} }
\label{fig:mini_httpd_dedotdot}
\end{figure}

\subsubsection{Canonicalizing Path Strings}
Functions that canonicalize path strings such as coreutils’ realpath \cite{coreutils} or the PHP programming language’s realpath \cite{phprealpath} resolve a user-supplied path string to a unique, absolute path string by expanding symbolic links and resolving directory references such as "/./", "/../", and extra "/" characters \cite{coreutils, notgethacked}. This means that even though there are theoretically infinitely many ways to refer to the same file--e.g. "/etc/passwd", "/../etc/passwd", "/../../etc/passwd", "//etc//passwd"…-- after canonicalization, they all equal the same path string. 

These canonicalization functions are effective for those who need to actually canonicalize path strings but overkill for those that only need to prevent directory traversal attacks. Since these functions resolve symbolic links, they must make system calls or at least access the filesystem \cite{symlinks}. This means that the underlying implementation of these functions is often more complex than other approaches and less portable. For example, porting PHP code using realpath to another programming language or web application framework means rewriting the realpath function.

\subsubsection{Whitelisting}
Whitelisting is a very effective approach to preventing directory traversal attacks but has some problems when used alone \cite{whitelisting}. Mentioned before, there are infinitely many ways to refer to the same file on a file system given that the length of the path string is not bounded. While effective in some scenarios, whitelist-only approaches can be overly strict if they do not canonicalize and allow the infinitely many other path strings that refer to a file. Since the specification of whitelisted files often happens during development, it could be difficult for these mitigations to securely adapt to circumstances where the whitelist needs to be updated dynamically. Further, if whitelisted directories contain symbolic links pointing to non-whitelisted directories or files, attackers can abuse these symbolic links to bypass defenses and traverse the file system.

\section{Algorithm}
The main algorithm, shown in Algorithm 1 and named $sanitize$ for ease of reference, is simple and intuitive. The algorithm takes a user-supplied path string and a list of whitelisted path strings as input. It returns true if the canonicalized user-supplied path string is in the whitelisted path string list and false otherwise. An example of an implementation written in the C programming language is available on GitHub \cite{mycode}.

\SetKwInput{KwInput}{Input}
\SetKwInput{KwOutput}{Output}

\begin{algorithm}[]
\caption{$sanitize$}

\KwInput{$userSuppliedPathString$ : $string$}
\KwInput{$whitelistedPathStrings$ : $string array$}
\KwOutput{$true$ or $false$}
\BlankLine
$stack\gets$ a stack initialized to be empty; \\
$tokens\gets$ tokenize($userSuppliedPathString$);
\BlankLine

\For{$i\gets 0$ to $tokens.size - 1$}{
    $currentToken\gets$ $tokens[i]$;
    \If{$currentToken = ".."$}{
        \eIf{$stack.size =$ 0}{
            $continue$;
        }{
            $stack.pop$;
        }
    }
    \uElseIf{$currentToken = "."$}{
        $continue$;
    }
    \Else{
        $stack.push(currentToken)$;
    }
}
\BlankLine
\eIf{$whitelistedPathStrings.contains(stack.toString())$}{
    return $true$;
}{
    return $false$;
}

\end{algorithm}

\begin{algorithm}[]
\caption{toString() method of stack}

$s \gets$ stack object

$result \gets$ "/" \newline

\If{$s.size > $ 0}{
    \For{$i \gets $ 0 to $s.size$}{
        $result.append(s.data[i])$;
        \BlankLine
        \If{$i \neq (s.size - 1)$}{
            $result.append("/")$;
        }
    }
}
\BlankLine
return $result$;

\end{algorithm}

We start by initializing an empty stack object that only needs to support the common push, pop, and size methods as well as one extra method shown in Algorithm 2: $toString$. We tokenize the user-supplied path string on all substrings composed of single or repeated occurrences of the "/" character. For the proof-of-concept C code we provide, we tokenize the user-supplied path string using the $strtok\_r$ function of strings.h \cite{strtokr} with "/" supplied as the separator argument. After tokenization, each token represents a directory in the path string and is pushed to the stack. If a token with value "." is detected, it is ignored, and if a token with value ".." is detected we discard it and pop off the most recently pushed token to the stack.

We make the assumption that user-supplied path strings do not contain symbolic links or wildcard characters such as "?" or "*". If symbolic link functionality is necessary, implementers may add symbolic links to the list of whitelisted path strings instead of trying to modify $sanitize$ to resolve these symbolic links. This alleviates the need for the $sanitize$ function to call into the file system and removes any security holes created by a whitelist-only mitigation procedure.

\subsection{Proof of Correctness}
We start the proof of correctness off with some definitions to formally define and clarify phrases needed in later proofs.

\theoremstyle{plain}
\newtheorem{theorem}{Theorem}
\newtheorem{lemma}{Lemma}
\newtheorem{definition}{Definition}

\theoremstyle{definition}
\begin{definition}
A path string is a string that contains only valid filename characters and the "." and "/" characters and is of the form "/$file_1$/$file_2$/.../$file_n$" where $file_n$ is the filename of the target file. From this point forward, the "..." token in a path string stands for an ellipsis and represents a variable amount of directories and not a directory named "..." or a directory traversal character. 
\end{definition}

\theoremstyle{definition}
\begin{definition}
A path string, $S_1$, is a prefix of another path string, $S_2$, if and only if $S_1.length <= S_2.length$ and for all positive integers, $i, [i : 0 <= i <= S_1.length]$, $S_1[i] = S_2[i]$. We will later refer to stacks being prefixes of another string; this just means that if $K$ is a stack, then $K.toString()$ is a prefix of the other string.
\end{definition}

\theoremstyle{definition}
\begin{definition}
Path strings $S_1$ and $S_2$ are equivalent path strings if and only if $S_1.length = S_2.length$ and $S_1$ and $S_2$ are prefixes of each other. In other words, a typical string comparison of $S_1$ and $S_2$ will say that they are equal strings.
\end{definition}

\theoremstyle{definition}
\begin{definition}
A path string, $S_1$, contains a directory or filename, $file_n$, if and only if the string "/$file_n$/" is a substring of $S_1$. For example, the path string "/etc/passwd" contains the directory "etc" because "/etc/" is a substring of "/etc/passwd". It is not enough to check that the string name of the directory is a substring of the path string since "my\_dir" is a substring of "/my\_dir1/"even though the path string does not contain a directory named "my\_dir".
\end{definition}

\begin{lemma}
If the user-supplied path string is bounded in length, then the algorithm terminates.
\end{lemma}
\begin{proof}
If the supplied path string has length, $L$, then the loop runs at most $L$ times. Since operating systems normally set a maximum length of path strings, this max length is an upper bound on the number of possible loop iterations. For Linux operating systems, this length is typically 4096 characters.
\end{proof}

\begin{lemma}
If $S_1$ and $S_2$ both refer to the same file on the filesystem, and--without loss of generality--$S_2$ is canonicalized and $S_1$ may or may not be canonicalized, then $S_1.length >= S_2.length$.
\end{lemma}

\begin{proof}
In the case that both $S_1$ and $S_2$ are canonicalized, $S_1.length = S_2.length$. However, if $S_1$ is not canonicalized, then $S_1$ must contain directories that $S_2$ does not contain or $S_1$ contains duplicate copies of directories contained in $S_2$ all of which must be popped off the stack by ".." directories. This follows from the assumptions that neither $S_1$ or $S_2$ contain symbolic links and that $S_1$ and $S_2$ both refer to the same file.
\end{proof}

\begin{lemma}
 If $S_1$ and $S_2$ both refer to the same file on the filesystem, and--without loss of generality--$S_2$ is canonicalized but $S_1$ may or may not be canonicalized, then at any iteration, $i$, of the algorithm’s loop, one of the following is true:
 \begin{enumerate}
     \item The stack of $S_1$ is a prefix of $S_2$, or
     \item If a directory is pushed to $S_1$'s stack that makes $S_1$'s stack become not a prefix of $S_2$, then this directory must eventually be popped off the stack by a following ".." directory.
 \end{enumerate}
\end{lemma}

\newtheoremstyle{case}{}{}{}{}{}{:}{ }{}
\theoremstyle{case}
\newtheorem{case}{Case}

\begin{proof}
Path string, $S_1$'s, stack starts as "/" which is a prefix of all possible path strings since all path strings are absolute path strings, so $Condition 1$ is met before the loop starts. From here, there are two cases: the case where $S_1$ is canonicalized and the case where $S_1$ is not canonicalized.

 \begin{case}
 Path strings $S_2$ and $S_1$ are both canonicalized. Suppose $S_2$ is the path string, "/$a_1$/$a_2$/.../$a_i$". Then the first directory of $S_1$ pushed to the stack must be $a_1$ keeping $S_1$'s stack a prefix of $S_2$ since "/$a_1$/" is a prefix of "/$a_1$/$a_2$/.../$a_i$". Any next token keeps $S_1$'s stack a prefix of $S_2$, so $Condition 1$ holds. At the end of the loop, the string representation of $S_1$'s stack is equal to the path string $S_2$ meaning that $S_1$ is still a prefix of $S_2$ and that the invariant still holds.
 \end{case}
 
 \begin{case}
 Path string $S_2$ is canonicalized, but path string $S_1$ is not canonicalized. 
We use the coloring scheme presented in Fig. 4 to visualize path strings. A directory is green if it is contained in both $S_1$ and $S_2$ in its correct place in the path string. A directory is red if it is not contained by both $S_1$ and $S_2$ or if it is in the incorrect place in the path string.
 \end{case}
 
  If $S_1$ and $S_2$ both refer to the same file and there are no symbolic links in either, then any red directory that is added must be removed by a matching ".." directory afterward as shown by Fig. 4. By contradiction, if any of these red directories that make $S_1$ become not a prefix of $S_2$ are not popped off, then there is a red directory in the path string of $S_1$ meaning that $S_1$ and $S_2$ could not possibly refer to the same file.

\end{proof}

\begin{figure}[htp]
    \centering
    \includegraphics[width=9cm]{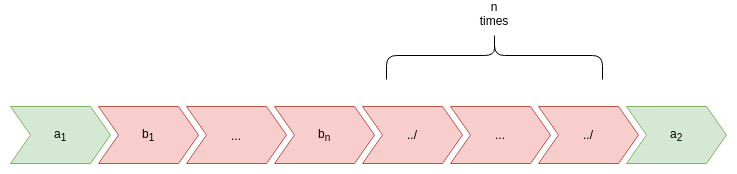}
    \caption{Example of $sanitize$ processing a path string
of which $condition 2$ of $Lemma 3$ holds true.}
    \label{fig:second_case_invariant}
\end{figure}

\begin{figure}[htp]
    \centering
    \includegraphics[width=4cm]{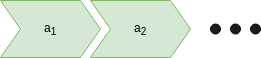}
    \caption{Example of $sanitize$’s stack after
processing/popping all red tiles off.
}
    \label{fig:after_second_case_invariant}
\end{figure}

\begin{theorem}
Assuming that no path string contains symbolic links and that whitelisted path string $S_2$ is canonicalized while user-supplied path string $S_1$ may or may not be canonicalized, then $sanitize(S_1, [S_2])$ is true if and only if path strings $S_1$ and $S_2$ refer to the same file on the file system.
\end{theorem}

\begin{proof}
$ $\newline \newline
$( \Rightarrow ):$ If $sanitize(S_1, [S_2])$ is true, then path strings $S_1$ and $S_2$ refer to the same file on the file system. Since $sanitize(S_1, [S_2])$ returned true, $[S_2].contains(S_1)$ must have returned true. Since the whitelisted path string array, $[S_2]$, contains only path string $S_2$, $S_2$ must be equal to $S_1$.
\newline \newline
$( \Leftarrow ):$ If path strings $S_1$ and $S_2$ refer to the same file on the file system, then $sanitize(S_1, [S_2])$ is true. We work backward to prove this starting at the last iteration of the for loop in the algorithm. We can choose to start at the last iteration because we know by $Lemma 1$ that the algorithm terminates. We know by $Lemma 3$ that one of the two statements in the invariant must hold.

Since we are on the last iteration of the loop and processing the last token, $A_n$, then one of three things is true: \newline

\begin{enumerate}
    \item $A_n$ belongs to $S_2$,
    \item $A_n$ is the token "..", 
    \item or, $A_n$ is the token ".".
\end{enumerate}
$ $\newline
We first prove by contradiction that $A_n$ is one of the above claimed tokens. We suppose that we are on the last iteration of the loop and the last token, $A_n$, is a directory that does not belong to $S_2$ which is the fourth and only other option. However, if $A_n$ doesn’t belong to $S_2$ and $S_1$ and $S_2$ both refer to the same file, then by the red-green example from $Lemma 3$, we must process an additional ".." directory to pop this directory off. Therefore, we could not be on the last iteration of the loop. 

Now we prove that the invariant holds on the three possible cases/token values of $A_n$ proposed before:

\begin{case}
$A_n$ belongs to $S_2$. Since we are processing the last token and $S_1$ and $S_2$ refer to the same file, then $S_1$ remains a prefix of $S_2$. Therefore, the string version of $S_1$'s stack is equal to path string $S_2$, so $sanitize(S_1, [S_2])$ must return true since $[S_2].contains(S_1)$ returns true.
\end{case}

\begin{case}
$A_n$ is the directory ".." If this is the case, then the previous token, $A_{n-1}$ is a red directory not contained by $S_2$. Therefore, the current ".." token must pop this red directory off causing $S_1$ to once again become a prefix of $S_2$. Since this is the last iteration, the string version of $S_1$'s stack is equal to the path string $S_2$ minus the additional “wrong” directory at the end of $S_1$'s stack string. Popping off this last directory makes $sanitize(S_1, [S_2])$ true.
\end{case}

\begin{case}
$A_n$ is the directory ".". If we have a path string, $S_n$, adding a "/./" directory makes $S_n$ still refer to the same file by the $sanitize$ algorithm because "/./" tokens are ignored. This means that the previous token kept $S_1$ a prefix of $S_2$ and the string version of $S_1$'s stack equal to $S_2$, and since "." is ignored, this remains true after this last iteration. Thus, $sanitize(S_1, [S_2])$ returns true.
\end{case}

\end{proof}

Finally, since we use the contain smethod for arrays/lists to check if the path string is whitelisted, the default behavior of $sanitize$ is to disallow all path strings. Then, if the path string, $S_2$, is added to the whitelisted path strings array, then by the theorem proved above, only path strings referring to the same file as $S_2$a re allowed.

\section{Testing and Verification Strategy}

One of the largest advantages and motivations of the $sanitize$ algorithm is its ease of verification. We will show two example verification strategies by first enumerating all possible inputs that lead to $sanitize$ outputting a particular string \cite{SurveySymExec-CSUR18} and second by showing that no traversal characters are left in the output string. For verification, we run KLEE symbolic execution engine \cite{Cadar:2008:KUA:1855741.1855756} on a machine with a dual-core Intel i5-7200U CPU and 8 GB of RAM. The code provided on GitHub \cite{mycode} contains code with all of the tweaks necessary for using KLEE and all output from each symbolic execution run described below.

For ease-of-use, we use a slightly modified version of $sanitize$ that returns the canonicalized copy of the input string instead of a boolean. We also substitute directory and file names in the path string by single-letter names to avoid path explosion problems \cite{SurveySymExec-CSUR18}. For example, we represent the path string "/home/user/.ssh/" by the new path string "/a/b/c/", with "a" representing "home", "b" representing "user", and "c" representing ".ssh". By doing this, we can also get away with symbolically executing $sanitize$ on a much shorter string. In our case, we will use strings with length at most 12 characters. We mark this input string to $sanitize$ as symbolic input using the $klee\_make\_symbolic$ function and tell KLEE to make the last character in the string, path\_string, a null terminator by using the $klee\_assumefunction$. 

For our first example, we can enumerate all input path strings that result in $sanitize$ returning the path string "/a/b/c" by adding the following code to the end of $sanitize$: \newline
$ $\newline
$if (!strncmp ("/a/b/c", canonicalizedString, 6)) \newline 
klee\_assert (0);$ \newline
$ $\newline
and passing KLEE the -emit-all-errors flag. We run KLEE on this modified $sanitize$ code, and the verification runs in about 137.23 seconds. After this run, KLEE presents us with 1142 unique input path strings that cause $sanitize$ to return the path string, "/a/b/c". Looking at some of the 1142 input path strings that result in this path string, most are just the same path string, "/a/b/c", with variable amounts of the "/" character separating the directories--e.g. "//a//b/c", "a/b/c////", "//a//b///c", etc. This is caused by our usage of the c function $strtok\_r$ for tokenizing path strings which can split on contiguous sequences of the desired token character. Other path strings have extra "/./"and "/../" tokens throughout that still keep the output path string equivalent to the input path string--e.g. "a/.//b/c", "../a/b/c", "./a/b/c", etc. By a quick sweep through all 1142 input path strings that result in an output of "/a/b/c", we can see that all of these strings are safe to be allowed through when "/a/b/c" is whitelisted and that there are no semantic issues with $sanitize$.

To verify that no directory traversal characters can bypass our algorithm, we switch out the call to $strncmp$ in our first example for three new calls to $klee\_assert$ where we make one call each to make sure that each of "/./", "/../",and "//" do not appear in the output string. We then compile this new code and run KLEE on it which takes about 323.00 seconds total. Since KLEE emits no errors, this means that of all the possible strings of length at most 12 characters, none can contain directory traversal characters after being ran through $sanitize$.

Because of how quick KLEE runs on the algorithm for path strings of length at most 12 characters, verification of the algorithm is quick enough for implementers to utilize during short development cycles or for regression testing. And while 12 characters might not seem like enough to test most path strings, substituting directory and file names with single character directory names like in the above example is an effective and equivalent way to represent longer path strings and to reduce run time of the symbolic execution engine.

\section{Advantages, Concerns, and Concluding Remarks}

\subsection{Advantages}

There are many advantages of $sanitize$ over other directory traversal defenses. The algorithm is lightweight, easy to implement, easy to verify, self-contained, and easily extendable. While we implemented the algorithm in the C programming language and had to provide our own stack implementation, many programming languages used in popular web application frameworks--e.g. Java, Ruby for Ruby-On-Rails, Python for Django and Flask, etc.--provide their own standard library that includes stacks or other containers implementing the stack interface. This means that the implementation of $sanitize$ could be reduced down to one or two simple functions. Because the algorithm is so simple and there are no calls to any external libraries or file systems, developers of most any skill level can port the algorithm to their language of choice.

Furthermore, since we have formally proved the correctness of $sanitize$, the only burdens an implementer must take on are making sure that their implementation closely follows the algorithm and that there are no other issues such as memory-safety issues, wrongly whitelisted path strings, etc.

The algorithm’s functionality is also simple to extend if an implementer desires any additional functionality. While potentially lowering the security posture of the algorithm, implementers could extend the algorithm to allow special characters such as wildcards for ease-of-use. Another extension is restricting all file operations to the current 'root' of a running web application; in this case, extending the algorithm to prepend the web application root directory’s absolute path to the user-supplied path string before $sanitize$ is ran is simple.

\subsection{Concerns}
While there are many advantages to using $sanitize$, the algorithm does have some disadvantages. For example, implementations of $sanitize$ in non-memory-safe programming languages require large amounts of unsafe operations such as string copies and memory manipulation. Coupling our algorithm with whitelisting presents a solid defense against directory traversal attacks but opens implementers up to other problems such as buffer overflows, out-of-bounds reads and writes, etc. depending upon the implementation. However, this flaw could also be viewed as an advantage, since fuzz-testing for memory-safety issues is much simpler than proving that others’ approaches to canonicalization are semantically correct. 

Another disadvantage is that $sanitize$ does not process path strings in the same way as most operating systems or file systems. The most notable side effect of this is that path strings that the file system may view as invalid can be marked as valid by $sanitize$. Take for example the path string, "/home/NonexistentUserFolder/../ActualUserFolder/". Our algorithm will canonicalize this path string to "/home/ActualUserFolder/" which does exist. However, an approach that uses the file system will try to traverse through the first non-existent user directory, fail immediately, and say that the overall path string or file does not exist.

A similar case results when $sanitize$ encounters path strings ending with the "/./" directory. While file systems only allow this "/./" directory to refer to directories--i.e. "/etc/./" is valid because "/etc/" is a directory, but "/etc/passwd/./" is not valid because "/etc/passwd" is a file--$sanitize$ allows the latter path string as the "/./" is simply ignored in the algorithm’s loop. But since the whitelisted path strings should only contain valid directories and files, these improper directory traversal semantics do not pose problems.

\subsection{Concluding Remarks}

In this paper, we presented a new and simple algorithm for preventing directory traversal attacks through canonicalization of path strings. We analyzed the strengths and weaknesses of some other common mitigations for these attacks and discussed how our algorithm is flexible and lightweight while still remaining more secure than other similar solutions. Finally, we provided a proof of correctness and some simple verification strategies using symbolic execution--available along with implementation code on GitHub--to prove that our algorithm prevents directory traversal attacks when properly used. 

\bibliographystyle{./bibliography/IEEEtran}
\bibliography{./bibliography/references}

\end{document}